\documentclass[times,11pt]{ettauth}

\usepackage{moreverb}

\usepackage[colorlinks,bookmarksopen,bookmarksnumbered,citecolor=red,urlcolor=red]{hyperref}

\newcommand\BibTeX{{\rmfamily B\kern-.05em \textsc{i\kern-.025em b}\kern-.08em
T\kern-.1667em\lower.7ex\hbox{E}\kern-.125emX}}

\usepackage[toc,page]{appendix}
\newtheorem{thm}{Theorem}
\newtheorem{cor}{Corollary}

\newtheorem{remark}{Remark}
\usepackage{color,soul}
\begin{document}

\title{Employing Antenna Selection to Improve Energy-Efficiency in Massive MIMO Systems}

\author{Masoud Arash\corrauth, Ehsan Yazdian$^*$, Mohammad Sadegh Fazel$^*$, Glauber Brante$^\dagger$, Muhammad Imran$^\ddagger$}

\address{$^*$ Isfahan University of Technology (IUT), Isfahan, Iran \\
$^\dagger$ Federal University of Paraná, Curitiba, Brazil\\
$^\ddagger$ University of Glasgow, Glasgow, UK}

\corraddr{Masoud Arash, Isfahan University of Technology (IUT), Isfahan $8415683111$, Iran. \\ E-mail:M.arash@ec.iut.ac.ir}

\begin{abstract}
Massive MIMO systems promise high data rates by employing large number of antennas, which also increases the power usage of the system as a consequence. This creates an optimization problem which specifies how many antennas the system should employ in order to operate with maximal energy efficiency. Our main goal is to consider a base station with a fixed number of antennas, such that the system can operate with a smaller subset of antennas according to the number of active user terminals, which may vary over time. Thus, in this paper we propose an antenna selection algorithm which selects the best antennas according to the better channel conditions with respect to the users, aiming at improving the overall energy efficiency. Then, due to the complexity of the mathematical formulation, a tight approximation for the consumed power is presented, using the Wishart theorem, and it is used to find a deterministic formulation for the energy efficiency. Simulation results show that the approximation is quite tight and that there is significant improvement in terms of energy efficiency when antenna selection is employed.
\end{abstract}


\maketitle

{}

\section{Introduction}
Massive MIMO systems, which employ a very large number of antennas, is one of the promising technologies for fifth-generation (5G) of wireless communication systems~\cite{hoydis2013massive}. Such transmission techniques are capable of focusing energy into very small regions of space, bringing huge improvements in terms of throughput and energy efficiency (EE). Moreover, the large number of antennas may also help the next generation of wireless systems to manage the growing number of user terminal (UTs), with possible high data rates demands, as is the case for Internet of Things deployments\cite{hoydis2013massive,andrews2014will}. In addition, massive MIMO may also simplify other parts of a communication system, such as the resource allocation and multiple access layer~\cite{larsson2014massive}. 

On the other hand, worries about global warming and greenhouse gas emissions result in limitations in the total power consumption of the communication systems~\cite{imran2011energy}. This brings an interesting trade-off between the system performance in terms of data rate serving multiple UTs due to the increased number of antennas~\cite{hoydis2013massive}, and the increased radiated power and power consumption~\cite{li2014energy}. Such a trade-off reflects directly into the EE, which is a parameter that gives a general view of the amount of optimality of a communication system, by encompassing the efficiency of the consumed power relative to the offered data rate~\cite{ngo2013energy}. In some works, due to mathematical complexity of EE, an approximation of it is used. For example, in~\cite{zhao2014performance} an approximation for EE has been optimized with respect to number of UTs.

In order to reduce the transmit power, some works in the literature proposed power control strategies. For instance, in~\cite{sifaou2016polynomial} the authors presented a simplified approach to reduce the complexity of deploying optimal linear precoders, with the goal of ensuring targeted UTs' data rates. In a multi cell massive MIMO scenario, \cite{van2016joint} optimized power allocation constrained with a total transmit power budget. In addition, the joint allocation of power and pilot symbols have also been considered, trying to mitigate pilot contamination in massive MIMO in~\cite{chien17pilot}, or by using the spectral efficiency as performance metric and setting a total energy budget per coherence interval in~\cite{cheng17pilot}. Also, in~\cite{guo2015energy} the authors proposed a joint pilot and data power control scheme to optimize EE subject to each UT's signal-to-interference plus noise ratio (SINR) and power constraints.

Moreover, an important factor that cannot be neglected in the analysis is the power used by the RF chains and signal processing of the transceivers, especially when the number of antennas increases. Neglecting the power consumption of RF chains causes a monotonic increase of the EE relative to the number of antennas. However, as highlighted by~\cite{mukherjee2014energy, bjornson2014optimal}, EE will not necessarily increase by employing more antennas in the case of massive MIMO systems, which results in an EE optimization problem with respect to the number of antennas and many other system parameters such as the number of UTs, the desired data rate, or other analog devices and residually lossy factors~\cite{ha2013}. Then, trying to tackle this problem, the authors in~\cite{bjornson2014optimal} have presented a comprehensive model for the power consumption of a communication system that can be used to accurately analyze the EE of massive MIMO systems.

Furthermore, the complexity of the MIMO transmission schemes may also be an important limitation in some scenarios. Therefore, antenna selection has been adopted in conventional MIMO systems~\cite{ghrayeb2006survey}, whose idea is to employ only a subset of all available transmit and receive antennas \cite{ghrayeb2006survey} in order to reduce the required number of RF chains, and consequently the energy consumption, as well as the system complexity. Moreover, depending on the optimization goal, the subset of antennas may be formed using different criteria, such as EE, throughput, complexity or \emph{etc}. Recently, the idea of antenna selection has also been analyzed in the massive MIMO context. For instance, the distribution of the mutual information between transmit and/or receive antennas has been derived for a point-to-point massive MIMO communication system employing transmit antenna selection in~\cite{hesami2011limiting, li2014energy}. However, there is a lack of comprehensive analysis for the effects of antenna selection in the Multi-User (MU) case, which has the greatest potential to be employed in the next generation of wireless systems~\cite{larsson2013massive,andrews2014will}. As it has been depicted by~\cite{li2014energy}, the big challenge lies in finding an exact relation for the mutual information with antenna selection in the MU case. 

In addition, with different approaches to improve the EE, the authors in~\cite{fang2015raise} proposed an algorithm to select a subset of antennas that results in the maximum determinant of the channel matrix. Selecting this submatrix has been shown to be the optimum choice from the capacity point of view in MIMO systems~\cite{molisch2005capacity}, however, the complexity of this method is still very high in massive MIMO systems due to the large number of antennas. In \cite{bai2015energy}, a capacity maximizing algorithm is used to improve EE for two scenarios: fixed or variable power consumption. However, in first scenario complexity of the algorithm is very high in massive MIMO systems and becomes even more when power consumption is assumed to be variable with data rate. Moreover, a limited number of RF chains (much less than the total number of available antennas) has been considered in~\cite{molisch2005capacity,li2014energy,dong2011adaptive}. Then, the idea is to cycle the RF chains among the whole set of antennas, which is suitable for situations where channel is slow changing enough. The most important limitation of this approach in a massive MIMO framework is that the diversity is limited by the number of RF chains, and not by the number of available antennas. Note that, in this case, number of selected antennas must always be less than the number of available RF chains. Also, another drawback comes in obtaining channel state information (CSI) of all antennas~\cite{ngo2013energy}, since the cycling of the RF-chains requires more time (or less pilot bits) to achieve full CSI, implying in imperfect or outdated estimation. Therefore, it may be more beneficial to employ the same number of RF-chains as the number of antennas, and then turn off some RF-chains/antennas whenever possible, to operate in a more energy efficient mode.

In this paper, we present an algorithm for antenna selection for MU massive MIMO systems, which is cost-effective related to the system dimensions. In order to reduce the complexity of mathematically analysis of EE, we first propose a selection algorithm based on the norm of the channel coefficients, which is also able to provide a good performance in terms of EE. Then, based on the characteristics of the proposed algorithm, we approximate the EE formulation using ordered statistics, which is shown to be very tight with the optimal case. Our results show a significant improvement in terms of EE when using only a subset of antennas. In general, the optimal number of antennas is always slightly higher than the number of UTs, but the relation depends on the required data rate. If the required data rate is too high (in the order of $100$~Mbps in our examples), then the algorithm opts to use the complete set of antennas. Moreover, when both the number of employed antennas and the average data rate to the users are optimized, up to $110\%$ EE improvement is obtained based on our system assumptions.

The rest of the paper is organized as follows: In section~\ref{System Model} system model is presented. Power consumption model and EE is formulated in Section~\ref{Energy Efficiency}. In Section~\ref{Antenna Selection} our selection algorithm is proposed, with the calculation of the approximated EE given in Section~\ref{Estimation}. Moreover, in Section~\ref{Numerical Results} numerical results will qualify our analysis and show the improvement of EE, achieved due to antenna selection, while the conclusions of this work are given in Section~\ref{Conclusion}.

\emph{Notation}: Boldface (lower case) is used for column vectors, $\boldsymbol{x}$, and (upper case) for matrices, $\boldsymbol{X}$. $\boldsymbol{X}^H$ and $\boldsymbol{X}^T$ illustrates conjugate transpose and transpose of $\boldsymbol{X}$, respectively. $\mathbb{E}\{.\}$ denotes expectation, $\mid.\mid$ stands for absolute value of a given scalar variable, $\parallel.\parallel$ is Euclidean norm, $tr$ is trace operator, $(x)^+=\max(x,0)$ and $\boldsymbol{X}_{(k,k)}$ represents $(k,k)^\text{th}$ entry of $\boldsymbol{X}$. Also, $\boldsymbol{I}_K$ is $K\times K$ identity matrix.

\section{System Model} \label{System Model}
We consider both uplink and downlink of a cellular system, with a single cell consisting of $K$ single antenna UTs and a massive MIMO base station (BS) at the center of the cell. Moreover, we assume that the BS selects $F$ out of $M$ available antennas to serve the UTs. Operational mode is assumed to be Time-Division-Duplex (TDD), so that the CSI acquired at the BS can be used in both uplink and downlink transmissions, due to channel reciprocity. 

Then, at the BS in the uplink, the received signal vector $\boldsymbol{y}^{(ul)} \in \mathbb{C}^{F\times 1}$ can be expressed as
\begin{equation}
\boldsymbol{y}^{(ul)}=\tilde{\boldsymbol{G}}\boldsymbol{Q}^{(ul)}\boldsymbol{s}+\boldsymbol{n}^{(ul)}, \label{recived signal}
\end{equation}
where $\tilde{\boldsymbol{G}}$ is the $F\times K$ channel matrix between $F$ selected antennas in BS and $K$ UTs, $\boldsymbol{Q}^{(ul)}$ is a $K\times K$ diagonal matrix, whose $(k,k)^\text{th}$ element is the square root of the $k^\text{th}$ UT's allocated power, $\boldsymbol{s}\in\mathbb{C}^{K\times 1}$ is the vector of transmitted symbols, with $\mathbb{E}\{\boldsymbol{s}\}=\boldsymbol{0}$ and $\mathbb{E}\{\boldsymbol{s}\boldsymbol{s}^H\}=\boldsymbol{I}_K$, and $\boldsymbol{n}^{(ul)}\sim \mathcal{CN}(0,\sigma_n^2\boldsymbol{I}_F)$ is additive white Gaussian noise. $\boldsymbol{y}^{(ul)}$ is then multiplied by combining matrix, $\boldsymbol{U}$.

In the downlink, by considering that the BS employs a precoding matrix $\boldsymbol{V}^T$ in order to simplify detection of the transmitted symbols at UTs, the signal received by the UTs can be written as
\begin{equation}
\boldsymbol{y}^{(dl)}=\tilde{\boldsymbol{D}}\boldsymbol{V}^T\boldsymbol{Q}^{(dl)}\boldsymbol{s}+\boldsymbol{n}^{(dl)}, \label{yDL}
\end{equation}
where $\tilde{\boldsymbol{D}}=\tilde{\boldsymbol{G}}^T$ is the downlink channel matrix, $\boldsymbol{Q}^{(dl)}$ is the downlink power allocation matrix and $\boldsymbol{n}^{(dl)}\sim \mathcal{CN}(0,\sigma_n^2\boldsymbol{I}_K)$ is the additive white Gaussian noise at the UTs.
Then, due to the dimensions of the system, it is more practical to use linear combiners such as Matched-Filter (MF), Minimum Mean-Squared Error (MMSE) or Zero-Forcing (ZF) due to their lower complexity~\cite{hoydis2013massive} compared to nonlinear methods such as DPC~\cite{jindal2005dirty}. Throughout this paper we assume a ZF processing at the BS, which allows us to achieve tractable mathematical expressions with performance close to MMSE~\cite{bjornson2014optimal} as it is shown in the Numerical Result section. Moreover, similar to~\cite{bjornson2014optimal}, we employ same precoder and combiner operators to reduce the computational complexity of the system, therefore $\boldsymbol{U}=\boldsymbol{V}$. In this case, the matrix $\boldsymbol{V}$ can be expressed as 
\begin{equation}
\boldsymbol{V}=(\tilde{\boldsymbol{G}}^H\tilde{\boldsymbol{G}})^{-1}\tilde{\boldsymbol{G}}^H.
\end{equation}
Consider $\boldsymbol{G}$ as a $M\times K$ matrix which models independent and identically distributed (i.i.d.) small-scale and large-scale fading channel, whose elements represent the channel between each of the $M$ antennas of the BS and the each of $K$ UTs, so that we express it as
\begin{equation}
\boldsymbol{G}=\boldsymbol{H}\boldsymbol{B}^\frac{1}{2}, \label{GHB}
\end{equation}
where
\begin{equation} \label{eq:H}
\boldsymbol{H}(m,k)\sim \mathcal{CN}(0,2\sigma^2),
\end{equation}
models the small-scale fading coefficient between the $k^\text{th}$ UT, $k \in [1,K]$, and the $m^\text{th}$ antenna at the BS, $m \in [1,M]$. Let us remark that we assume perfect CSI availability in the BS side. Also, the variance of each real and imaginary part of $\boldsymbol{H}(m,k)$ is equal to $\sigma^2$. $\boldsymbol{B}$ is a diagonal $K\times K$ matrix, whose $(k,k)^\text{th}$ element represents the $k^\text{th}$ UT's path-loss coefficient, $r_k$. Since the BS antennas are usually placed in much smaller dimensions relative to the distance between UTs and the BS, it is sensible to assume that $r_k$ is the same for all BS antennas. Note that the matrix $\tilde{\boldsymbol{G}}$ in Eq.~\eqref{recived signal} is the channel matrix using only the $F$ selected antennas ($F\leq M$), thus it is obtained using $F$ appropriately selected rows from $\boldsymbol{G}$, as will be detailed in Section~\ref{Antenna Selection}.

\section{Energy Efficiency} \label{Energy Efficiency}
Assume $R_k$ (in bit/channel use) as the average achievable information rate of the $k^\text{th}$ UT and $P_\text{tot}$ as the total power consumption in transceivers to achieve sum rate of all UTs. Then, according to~\cite{mukherjee2014energy}, the average EE is defined as
\begin{equation}
\eta_\text{E}=\frac{\sum_{k=1}^{K}R_k}{P_\text{tot}}. \label{EE def}
\end{equation}
In this paper we assume that all UTs achieve an equal gross rate, $R_k=\bar{R}, k=1, \cdots , K$. To do so, we employ the power allocation scheme presented in~\cite{wiesel2006linear}, which guarantees a constant rate to the UTs, and which provides closed-form results for the consumed energy. Therefore, despite different channel conditions, all UTs will serve with the same rate $\bar{R}$ and therefore the power consumption of each UT will be a function of $\bar{R}$ and its channel condition. Following~\cite{wiesel2006linear} for instantaneous uplink power allocation vector, $\boldsymbol{p}^{(ul)}$, we have 
\begin{equation}
\boldsymbol{p}^{(ul)}=\sigma_n^2 (\boldsymbol{A}^{(ul)})^{-1}\boldsymbol{1}_K,  \label{Aul}
\end{equation}
in which $\boldsymbol{A}^{(ul)}$ is defined as
\begin{equation} \label{Ak,k}
\boldsymbol{A}^{(ul)}_{(k,l)} =
\begin{cases}
\frac{1}{\parallel\boldsymbol{v}_k\parallel^2 (2^{\frac{\bar{R}}{BW}}-1)} &  \hspace{5mm} k=l,\\[1em]
0 &  \hspace{5mm} k\ne l,
\end{cases}
\end{equation}
where $\boldsymbol{v}_k$ is the $k^\text{th}$ column of matrix $\boldsymbol{V}^T$, $BW$ is the transmission bandwidth and $\boldsymbol{1}_K$ is $K\times 1$ vector whose all elements are equal to one. Note that $\boldsymbol{p}^{(ul)}(k)=(\boldsymbol{Q}^{(ul)}_{(k,k)})^2$ and $\boldsymbol{A}^{(ul)}$ is a diagonal matrix. Moreover, due to using the ZF scheme in both uplink and downlink, we also have that $\boldsymbol{A}^{(dl)}=(\boldsymbol{A}^{(ul)})^T$ and, as a result, the instantaneous power allocation at the downlink is the same as at the uplink, \emph{i.e.}, $\boldsymbol{p}^{(dl)} = \boldsymbol{p}^{(ul)}$ and $\boldsymbol{p}^{(dl)}(k)=(\boldsymbol{Q}^{(dl)}_{(k,k)})^2$~\cite{wiesel2006linear}.

It is worth noting that by employing this power allocation scheme one can specify an equal rate for all users, which can be constant or a function of other system requirements. For example, the authors in~\cite{bjornson2014optimal} considered a variable $\bar{R}$, which is a logarithmic function of number of BS antennas and number of UTs. Using this assumption the data rate of all UTs will be affected when a new UT is added, or disconnected from the network. Thus, in this paper we assume that all users have an arbitrary and equal rate $\bar{R}$, while we maintain our formulation for the general case.

We formulate the total power consumption by dividing it into three distinct parts: \emph{i.)} the emitted power, $P_\text{e}$, which depends on the coefficients of transmitted signal and is computed as the absolute transmitted power; \emph{ii.)} the power that the BS and the UTs consume to generate and process the transmitted signals, $P_\text{process}$, which is a function of number of UTs, number of operational antennas and processing algorithms for precoding and combining; and \emph{iii.)} $P_\text{fix}$, which is the constant power consumed by cooling systems, backhaul infrastructure, \emph{etc}. 

First, the average power consumption due to emitted power in the downlink can be calculated as \cite{bjornson2014optimal}
\begin{equation}
P_\text{e}^{(dl)}=BW\sigma_n^2 \, \mathbb{E}\{\boldsymbol{1}_K^T (\boldsymbol{A}^{(dl)})^{-1}\boldsymbol{1}_K\}, \label{Peup}
\end{equation}
where expectation is over both small-scale and large-scale fadings. Moreover, since the power allocation matrix in the uplink is the same as in the downlink in the case of the ZF scheme, we can conclude that $P_\text{e}^{(ul)} = P_\text{e}^{(dl)}$. Then, by assuming half of the coherence block for each uplink and downlink phases, the emitted power will be
\begin{equation}
P_\text{e}=\frac{1}{2}P_\text{e}^{(ul)}+\frac{1}{2}P_\text{e}^{(dl)}=P_\text{e}^{(dl)}.
\end{equation}
Notice that $P_\text{e}$ can be expressed analytically as a function of $K$ and $M$ by using the Wishart theorem as in~\cite{bjornson2014optimal}. However, as pointed out by~\cite{li2014energy}, it is too hard to find an exact expression for~\eqref{Peup} when antenna selection is employed. Nevertheless, in Section~\ref{Estimation} we derive an approximation for $P_\text{e}$ to be used in the case of antenna selection. 

Next, the power used to process the transmitted signals, $P_\text{process}$, can be decomposed into four main components as
\begin{equation}
\label{Pprocess}
P_\text{process}= P_\text{BB} + P_\text{RF} + P_\text{CSI} + P_\text{LP},
\end{equation}
where $P_\text{BB}$ is the power consumption due to baseband operations for coding and decoding, $P_\text{RF}$ is the RF circuitry power consumption, $P_\text{CSI}$ is used for channel estimation, and $P_\text{LP}$ encompasses linear precoding operations. Then, building up upon the model presented in~\cite{bjornson2014optimal}, we have that
\begin{equation}
P_\text{BB}=(P_\text{cod}+P_\text{dec})K, \label{Pcoding}
\end{equation}
\begin{equation}
P_\text{RF}=\frac{1}{2}\underbrace{(FP_\text{tx}+KP_\text{rx})}_\text{Downlink}+\frac{1}{2}\underbrace{(KP_\text{tx}+FP_\text{rx})}_\text{Uplink}, \label{Ptr}
\end{equation}
\begin{equation}
P_\text{CSI}=M\frac{K}{LT}, \label{Pest}
\end{equation}
\begin{equation}
P_\text{LP}=\frac{9K^2F+6KF+2K^3}{3LT}+(1-\frac{K}{T})\frac{FK}{L}, \label{Ppre}
\end{equation}
where $P_\text{cod}$ and $P_\text{dec}$ are constants, respectively denoting the power required for coding or decoding a symbol, while $P_\text{tx}$ and $P_\text{rx}$ are the constant powers used at transmitter and receiver RF-chains, respectively. Notice that we assume equal parts for downlink and uplink in each coherence block, so that a factor of $\frac{1}{2}$ appears for each term in~\eqref{Ptr}. Moreover, $L$ is the computational efficiency of the operations per Joule, and $T$ is the coherence time of channel. 

Rewriting \eqref{Pprocess} in terms of the dependence on $F$ or $K$ we have,
\begin{equation}
P_\text{process}=\sum_{i=1}^3 C_{i,0}K^i + \sum_{i=0}^2 C_{i,1}K^iF, \label{Psystem}
\end{equation}
in which $C_{1,0}=P_\text{cod}+P_\text{dec}+P_\text{rx}+\frac{M}{LT}$, $C_{2,0}=0$, $C_{3,0}=\frac{2}{3LT}$, $C_{0,1}=P_\text{tx}$, $C_{1,1}=\frac{3}{LT}+\frac{1}{L}$ and $C_{2,1}=\frac{2}{LT}$.

Then, to sum up, the total power consumption can be written as
\begin{equation}
P_\text{tot}=P_\text{e}+P_\text{process} + P_\text{fix}. \label{Ptot}
\end{equation}

Finally, assuming that all UTs are served with constant rate, $\bar{R}$, the EE can be written as
\begin{equation}
\eta_\text{E} =\frac{K\bar{R}}{P_\text{e}+\sum_{i=0}^3 C_{i,0}K^i + \sum_{i=0}^2 C_{i,1}K^iF+P_\text{fix}}, \label{EE}
\end{equation}
where notice that $P_\text{e}$, as defined in~\eqref{Peup}, is the sole term without a closed-form expression, which prevents us to find a deterministic form for the EE. Then, in the following we propose an algorithm for antenna selection and use its characteristics to provide a tight approximated formulation for $P_e$.

\section{Antenna Selection Algorithm} \label{Antenna Selection}
From both diversity order and multiplexing gain point of views, the system always benefits from increasing the number of antennas. However, a few trade-offs are observed in terms of the EE of massive MIMO systems. For instance, it has been shown in~\cite{bjornson2014optimal,mukherjee2014energy} that the EE tends to zero when the number of UTs is equal to the number of antennas ($K = M$) using the ZF precoder, or when $M \to \infty$. Thus, it can be concluded that the optimum EE may be achieved for some $F$ in $K<F\le M<\infty$, so that the system may employ an optimum number $F^\star$ out of the pool of $M$ available antennas. In the literature, some works such as~\cite{bjornson2014optimal} have optimized the EE with respect to $M$. However, in an operating massive MIMO system the number of antennas, $M$, is fixed. Therefore, to improve EE we propose to select a subset of antennas, still taking advantage of the massive MIMO transmission scheme, but consuming less energy. Then, since the mathematical analysis in antenna selection case is usually hard to be obtained~\cite{li2014energy}, our main goal is to provide a low-complex antenna selection algorithm, which also allows us to provide a closed-form approximation to the EE formulation.

In order to select the best set of antennas, several criteria have been proposed for conventional MIMO systems. For instance, the channel strength criterion results in maximum SNR; however, maximizing the SNR does not necessarily maximize the EE of the system~\cite{ghrayeb2006survey}. Another criterion that can be employed is to maximize the Frobenius norm, which has been shown in~\cite{gore2002mimo} to maximize SNR at the same time that it minimizes the instantaneous probability of error. Therefore, our approach is to select a subset of best antennas in the sense of maximizing the average absolute value of the channel coefficients. Notice that this method it not necessarily optimal in terms of EE, but its computational complexity is linear with number of UTs and BS antennas. 

The antenna selection algorithm is presented in Table~\ref{Table algorithm}. First we employ the absolute value operator to form the vector $\boldsymbol{a}$, whose $m^\text{th}$ element contains the sum of the absolute values of the channel coefficients between all UTs and the $m^\text{th}$ BS antenna. Note that we employ $\boldsymbol{H}$, as defined in~\eqref{eq:H}, rather than $\boldsymbol{G}$ to compute the elements of $\boldsymbol{a}$, since we are selecting the antennas based on their fading channel strength only. Next, in each step the maximum value of this vector is determined and the corresponding antenna is added to the subset $\boldsymbol{s}$, representing the selected antennas. Then, after each new inclusion in the subset $\boldsymbol{s}$, the EE is computed using~\eqref{EE} and compared to the previous EE. This process continues while the overall EE increases with the inclusion of new elements and the algorithm stops whenever EE at the step $n$ is lower than that at the step $n-1$. Finally, the algorithm returns the subset $\boldsymbol{s}$ of the selected antennas as well as $F = \text{card}\left(\boldsymbol{s}\right)$, where $\text{card}(.)$ represents the cardinality of the set.

\begin{table*}
	\centering
	\caption{Antenna selection algorithm}
	\begin{tabular}{l|l}
		\hline
		\textbf{Algorithm steps} & \textbf{Calculation order} \\\hline
		$\boldsymbol{a}:=[a_1 \hspace{1mm} a_2 \hspace{1mm} \ldots \hspace{1mm} a_M]$ & \\
		$\textbf{for} \hspace{1.5mm} m:=1 \hspace{1mm} \textbf{to} \hspace{1mm} M$&\\
		\hspace{6mm}$a_m=\sum_{k=1}^K\mid\boldsymbol{H}(m,k)\mid^2$ & $O(K)$ \\
		\textbf{end}&\\
		$\eta_\text{E}(0) := 0$ &\\
		$\textbf{for} \hspace{1.5mm} n:=1 \hspace{1mm} \textbf{to} \hspace{1mm} M $ &\\
		\hspace{6mm}$u:=\arg\max \, \boldsymbol{a}$ & $O(M)$ \\
		\hspace{6mm}$\boldsymbol{s}(n):=u$&\\
		\hspace{6mm}$\boldsymbol{a_u}:=0$&\\
		\hspace{6mm}compute $\eta_\text{E}(n)$ using~\eqref{EE} using only the subset $\boldsymbol{s}$&\\
		\hspace{6mm}\textbf{if} $\eta_\text{E}(n)<\eta_\text{E}(n-1)$&\\
		\hspace{12mm}delete $\boldsymbol{s}(n)$&\\
		\hspace{12mm}exit&\\
		\hspace{6mm}\textbf{end}&\\
		\textbf{end}&\\
		$F = \text{card}\left(\boldsymbol{s}\right)$  &\\
		\textbf{return} $F$, $\boldsymbol{s}$&\\
		\hline
	\end{tabular}
	\label{Table algorithm}
\end{table*}

Moreover, let us remark that the algorithm in Table~\ref{Table algorithm} requires the numerical computation of~\eqref{EE}, since the emitted power $P_\text{e}$ is not in closed-form. However, due to the use of the Frobenius norm, we are able to provide an approximation for $P_\text{e}$, and consequently simplify the antenna selection algorithm.

\subsection{Approximation of $P_e$ When Employing Antenna Selection}  \label{Estimation}
According to~\cite{bjornson2014optimal}, the Wishart theorem can be used to write $P_\text{e}$ as a function of $K$ and $M$ when all antennas are used. Let us denote by $\tilde{\boldsymbol{H}}$ the $F\times K$ matrix of selected small-fading coefficients, which is obtained in the same way as $\tilde{\boldsymbol{G}}$, \emph{i.e}, by using $F$ appropriately selected rows from $\boldsymbol{H}$. However, as pointed out by~\cite{li2014energy}, it is quite complex to solve~\eqref{Peup} in closed-form when only a subset of the $M$ antennas is used, since the the distribution of the elements of $\tilde{\boldsymbol{H}}$ is no longer Gaussian in the case of antenna selection. 

Then, in order to solve this problem, we propose to approximate $\tilde{\boldsymbol{H}}$ using an i.i.d. complex Gaussian matrix $\hat{\boldsymbol{H}}$, so that
\begin{equation}
\hat{\boldsymbol{H}}(f,k)\sim\mathcal{CN}(0,2\hat{\sigma}^2),\hspace{2mm} f \in [1,F],\hspace{1mm}k\in [1,K],\label{H_hat}
\end{equation}
where variance of each real and imaginary part of $\hat{\boldsymbol{H}}(f,k)$ is equal to $\hat{\sigma}^2$. Notice that this model has a single parameter to be determined, so that we only need to find $\hat{\sigma}^2$ analytically. Then, in order to achieve a fair approximation for $\tilde{\boldsymbol{H}}$, we first constrain the energies of $\tilde{\boldsymbol{H}}$ and $\hat{\boldsymbol{H}}$ to be equal, \emph{i.e.},
\begin{align}
\sum_{f=1}^{F}\sum_{k=1}^{K}\mathbb{E}\{\mid\tilde{\boldsymbol{H}}(f,k)\mid^2\}&=\sum_{f=1}^{F}\sum_{k=1}^{K}\mathbb{E}\{\mid\hat{\boldsymbol{H}}(f,k)\mid^2\}\nonumber\\&\stackrel{(\star)}{=}2FK\hat{\sigma}^2, \label{equalVar}
\end{align}
where $(\star)$ is concluded due to the independence among the elements of $\tilde{\boldsymbol{H}}$. Moreover, the left hand side of~\eqref{equalVar} can be written as
\begin{align}
\sum_{f=1}^{F}\sum_{k=1}^{K}\mathbb{E}\{\mid\tilde{\boldsymbol{H}}(f,k)\mid^2\}&=\sum_{f=1}^{F}\mathbb{E}\{\sum_{k=1}^{K}\mid\tilde{\boldsymbol{H}}(f,k)\mid^2\}\nonumber\\&=\sum_{f=1}^{F}\mathbb{E}\{a_{f:M}\}, \label{leftthandside}
\end{align}
where $a_{f:M}$ denotes the $f^\text{th}$ greatest value in $\boldsymbol{a}$ out of $M$, so that the sum of the $F$ largest values of $\boldsymbol{a}$ is required in~\eqref{leftthandside} to obtain $\hat{\sigma}^2$. In this case, due to antenna selection algorithm presented in Table~\ref{Table algorithm}, the mean of entries of $\boldsymbol{a}$ can be achieved using order statistics. 

Regarding the distribution of the matrix $\boldsymbol{H}$, we know that $|\boldsymbol{H}(f,k)|^2$ has an exponential distribution for all values of $f, k$, \emph{i.e.},
\begin{equation}
\mid\boldsymbol{H}(m,k)\mid^2\sim\mathcal{E}(2\sigma^2), \hspace{1mm} m=1,\ldots,M,\hspace{1mm}k=1,\ldots,K. \label{Hfkdis}
\end{equation}
Thus, according to Table~\ref{Table algorithm} each element of the $M\times1$ vector $\boldsymbol{a}$ is the sum of $K$ independent exponentially distributed random variables, which in turn results in a Gamma distributed random variable~\cite{papoulis2002probability}
\begin{equation}
a_m\sim\mathcal{G}(K,2\sigma^2),\hspace{5mm} m=1,\ldots,M. \label{a_idis}
\end{equation}

In the following step, we must choose the $F$ largest values of $\boldsymbol{a}$. To that end we follow~\cite{david2003order}, which has shown that for $M$ i.i.d. realizations of a random variable $z$, with probability density function (pdf) given by $P(z)$, and cumulative distribution function given by $C(z)$, we have the mean of the $r^\text{th}$ greatest realization, out of $M$, given by
\begin{align}
&\mu_{r:M}=\nonumber\\&\frac{M!}{(r-1)!(n-r)!}\int_{-\infty}^{\infty}zP(z)C^{r-1}(z)(1-C(z))^{(M-r)}dz, \label{muexact}
\end{align}
which can be achieved in closed-form for a few specific distributions, such as the Uniform distribution~\cite{david2003order}, but not for the case of the Gamma distribution. 

Nevertheless, by employing~\eqref{muexact} we can obtain the mean of the $F$ largest values in $\boldsymbol{a}$, and we resort to the following theorem to find an approximate value for $\sum_{f=1}^{F}\mathbb{E}\{a_{f:M}\}$. 
\begin{thm} \label{thm1}
	For jointly distributed random variables 
	$Z_1,Z_2,\ldots,Z_M$ with means $\mu_i=\mu$, and variances $\sigma_{iz}^2=\sigma_z^2$, $i=1,\ldots,M$, we have 
	\begin{equation} \label{thm1_1}
	\left\lvert\sum_{i=1}^{M}\lambda_i(\mu_{i:M}-\mu)\right\lvert\le\sigma_z\sqrt{M\sum_{i=1}^{M}(\lambda_i-\bar{\lambda})^2}, 
	\end{equation}
	for any $\boldsymbol{\lambda}=[\lambda_1,\ldots,\lambda_M]\in\mathbb{R}^M$, where $\bar{\lambda}=\frac{1}{M}\sum_{i=1}^{M}\lambda_i$.
\end{thm}
\begin{proof}
	Please see \cite{arnold2012relations}. 
\end{proof}
Assuming that the bound presented in Theorem~\ref{thm1} is tight \cite{arnold2012relations}, the following corollary can be employed to derive an upper bound for $\sum_{f=1}^{F}\mathbb{E}\{a_{f:M}\}$.
\begin{cor} \label{cor1}
	According to~\eqref{thm1_1}, the mean $\sum_{f=1}^{F}\mathbb{E}\{a_{f:M}\}$ is upper bounded by
	\begin{equation} \label{inequality1}
	\sum_{f=1}^{F}\mathbb{E}\{a_{f:M}\}\le2K\sigma^2+2\sigma^2\sqrt{K\frac{M-F}{F}}.
	\end{equation}
\end{cor}
\begin{proof}
	Please refer to Appendix~\ref{ApCor1}. 
\end{proof}

Due to the tightness of this upper bound, as the numerical results will show, in the rest of this paper we assume that
\begin{equation}
\sum_{f=1}^{F}\mathbb{E}\{a_{f:M}\}\approx 2K\sigma^2+2\sigma^2\sqrt{K\frac{M-F}{F}}. \label{approximation1}
\end{equation}
Then, putting~\eqref{approximation1} into~\eqref{equalVar}, $\hat{\sigma}^2$ is finally achieved as
\begin{equation} \label{sigmahat}
\hat{\sigma}^2=\sigma^2\left(1+\sqrt{\frac{M-F}{FK}}\right).
\end{equation}

Moreover, using this Gaussian approximation for the selected antenna matrix, we can also approximate $E\{tr(\tilde{\boldsymbol{H}}^H\tilde{\boldsymbol{H}})^{-1}\}$, which is needed to calculate $P_\text{e}$. Based on the Wishart theorem~\cite{tulino2004random}  we know that
\begin{equation}
\mathbb{E}\{tr(\hat{\boldsymbol{H}}^H\hat{\boldsymbol{H}})^{-1}\}=\frac{K}{F-K}\frac{1}{2\sigma^2(1+\sqrt{\frac{M-F}{FK}})}. \label{wishart}
\end{equation}
Then, since the matrix $\tilde{\boldsymbol{H}}$ is approximated by $\hat{\boldsymbol{H}}$,
\begin{align}
\mathbb{E}\{tr(\tilde{\boldsymbol{H}}^H\tilde{\boldsymbol{H}})^{-1}\}&\approx \mathbb{E}\{tr(\hat{\boldsymbol{H}}^H\hat{\boldsymbol{H}})^{-1}\}\nonumber\\&=\frac{K}{F-K}\frac{1}{2\sigma^2(1+\sqrt{\frac{M-F}{FK}})}. \label{trace}
\end{align}

As we can observe, the antenna selection algorithm introduces the term $(1+\sqrt{\frac{M-F}{FK}})$ to the denominator of~\eqref{trace}, comparing to the case when the system employs all $M$ antennas. This term indicates that there is an optimization to be performed in $F$, since $\mathbb{E}\{tr(\tilde{\boldsymbol{H}}^H\tilde{\boldsymbol{H}})^{-1}\}$, which is directly related to the transmitted power, is not an increasing function of $F$. This equation shows that $E\{tr(\tilde{\boldsymbol{H}}^H\tilde{\boldsymbol{H}})^{-1}\}$ decreases as the number of selected antennas $F$ increases; however, the speed of this decrease also depends on number of UTs, $K$. In addition, let us remark that the tightness of the approximation provided by~\eqref{trace} will be validated in Section~\ref{Numerical Results}. 

\begin{cor} \label{cor2}
	Using~\eqref{trace}, an approximation for the emitted power can be written as
	\begin{equation} \label{puldeter}
	\hat{P}_\text{e} = \frac{BW\sigma_n^2K\mathbb{E}\{r^{-1}\}(e^{\frac{\bar{R}}{BW}}-1)}{(F-K)2\sigma^2(1+\sqrt{\frac{M-F}{FK}})}, 
	\end{equation}
	where $r$ models the path-loss attenuation\footnote{Assuming that the large-scale fading is dominated by the path-loss.} and thus $\mathbb{E}\{r^{-1}\}$ is not a function of neither $F$ nor $K$.
\end{cor} 
\begin{proof}
	Please refer to Appendix~\ref{ApCor2}. 
\end{proof}

\subsection{Approximated Antenna Selection Algorithm} \label{approxAlg}
Using the approximation of Corollary~\ref{cor2}, we can simplify~\eqref{EE} so that it can be approximated in closed-form by
\begin{align} 
&\hat\eta_\text{E} =\nonumber\\ &\frac{K\bar{R}}{\frac{BW\sigma_n^2K\mathbb{E}\{r^{-1}\}(e^{\frac{\bar{R}}{BW}}-1)}{(F-K)2\sigma^2(1+\sqrt{\frac{M-F}{FK}})}+\displaystyle{\sum_{i=0}^3 C_{i,0}K^i} +\displaystyle{\sum_{i=0}^2 C_{i,1}K^iF}+P_\text{fix}}, \label{EEd}
\end{align}
so that the proposed simplified algorithm can be easily obtained from Table~\ref{Table algorithm} by replacing $\eta_\text{E}$ by $\hat\eta_\text{E}$, using \eqref{EEd}.

\begin{remark} \label{rem2}
	Assuming that the UTs are uniformly distributed in a circular cell with radius $d_\text{max}$, the path-loss attenuation is given by
	\begin{align} 
	r(d)= \frac{\bar{d}}{\parallel d\parallel^\kappa},
	\end{align}
	where $d$ is the distance between the user and the BS, $\kappa$ is the path-loss exponent and the constant $\bar{d}$ regulates the channel attenuation at a reference distance $d_\text{min}$. In this case, the $\mathbb{E}\{r^{-1}\}$ can be calculated as~\cite{bjornson2014optimal}
	\begin{align}
	\label{Er}
	\mathbb{E}\{r^{-1}\}= \frac{d_\text{max}^{\kappa+2}-d_\text{min}^{\kappa+2}}{\bar{d}\left(1+\frac{\kappa}{2}\right)\left(d_\text{max}^{2}-d_\text{min}^{2}\right)}.
	\end{align}
\end{remark}

Throughout this paper we assume that the UTs are uniformly distributed in a circular cell, so that $\mathbb{E}\{r^{-1}\}$ is given by~\eqref{Er}. This assumption allows us to provide a simple closed-form expression for $\hat\eta_\text{E}$, whose optimum point with respect to $F$ can be calculated using the Newton's method or with a simple one dimensional global search, which we employ in next section to provide a few numerical examples.

\begin{remark} \label{rem1}
	In addition to $F$, \eqref{EEd} also suggests that the system could benefit from the optimization of $\bar{R}$, \emph{i.e.}, the UTs data rate. Since $\bar{R}$ and $F$ are independent variables, one dimensional global search can also be employed to verify its effect on the maximization of the EE.
\end{remark}

\section{Numerical Results} \label{Numerical Results}
In this section we numerically evaluate the proposed algorithm, using the approximated approach depicted in Section~\ref{approxAlg}, as well as the Monte Carlo simulations for the algorithm in Table~\ref{Table algorithm}. The results are obtained employing 2000 Monte Carlo iterations. Moreover, we use the same parameters for power consumption as in~\cite{bjornson2014optimal}, which are summarized in Table~\ref{Table Parameter}. Also, the 3GPP distance-dependent path-loss model is used and the total number of available antennas is considered to be $M=220$ in all scenarios. All simulation results are provided for ZF case, unless otherwise stated.

\begin{table*}[!h]
	\centering
	\caption{Parameter values for simulations}
	\begin{tabular}{|c|c||c|c|}
		\hline
		\textbf{Parameter} & \textbf{Value} & \textbf{Parameter} & \textbf{Value}\\\hline
		Minimum UT distance: $d_\text{min}$ & $35$~m & Maximum UT distance: $d_\text{max}$ & $250$~m \\
		Path-loss at distance $d$: $r(d)$ & $\frac{10^{-3.53}}{r^{3.76}}$ &  Power for coding: $P_\text{cod}$ & $4$~W \\
		Coherence bandwidth: $BW$ & $180$~kHz & Power for decoding: $P_\text{dec}$ & $0.5$~W \\
		Coherence time: $T$ & $32$~ms & Power of the RF-chain at TX: $P_\text{tx}$ & $1$~W \\
		Noise variance: $\sigma_n^2$ & $10^{-20}$~J/c.u. & Power of the RF-chain at RX: $P_\text{rx}$ & $0.3$~W \\
		Channel coefficients variance: $2\sigma^2$ & $1$ & Constant power consumption: $P_\text{fix}$ & $18$~W \\
		Operations efficiency: $L$ & $10^9$/Joule & &  \\
		\hline
	\end{tabular}
	\label{Table Parameter}
\end{table*}

Fig.~\ref{fig1} shows the general behavior of the approximated proposed algorithm, using~\eqref{EEd} to obtain the optimum $\bar{R}$ using a global search. Therefore, the maximal $\hat\eta_\text{E}$ is illustrated for each $(F,K)$ pair. As we can observe, $\hat\eta_\text{E} = 0$ in the scenarios when $F < K$, so that the UTs cannot be served in this case, and the EE is maximized when the number of selected antennas is in the range $K \leq F \leq M$. The global optimum point, which is marked in Fig.~\ref{fig1}, is achieved when $F=137$ and $K=97$, yielding $\eta_\text{E}=28.40$~Mbits/J.

\begin{figure}[t]
	\centering
	\includegraphics[width=0.5\textwidth]{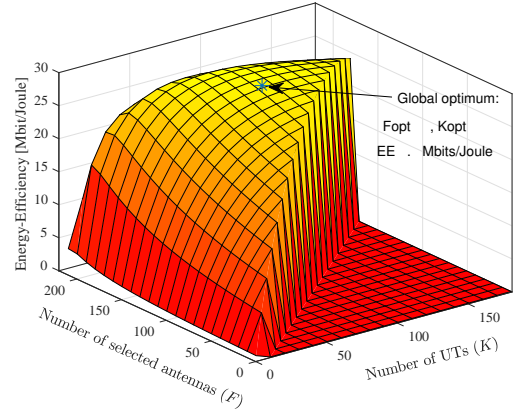}
	\caption{EE as a function of $K$ and $F$ using the proposed approximated algorithm, obtained employing~\eqref{EEd}. The global optimum point is also shown, for which $K=97$, $F=137$ and $\hat\eta_\text{E} = 28.40$~Mbits/J.} \label{fig1}
\end{figure}

Fig.~\ref{fig2} investigates the accuracy of Eq.~\eqref{puldeter} by comparing the Monte Carlo simulated $P_e$ with approximated one in~\eqref{puldeter}. In this figure, $\bar{R}/BW=5(bits/s/Hz)$ and simulation is done for $K=30$, $K=90$ and $K=150$. It is seen that the approximation for $P_e$ is quite tight and thus is a reasonable alternative for it. Therefore, since all other parameters are deterministic in Eq.~\eqref{EEd}, the proposed approximated EE also leads to very tight results with respect to the Monte Carlo simulations. Also,it is observed through simulations that as $F$ grows toward $M$, the approximation becomes more accurate.

\begin{figure}[t]
	\centering
	\includegraphics[width=0.5\textwidth]{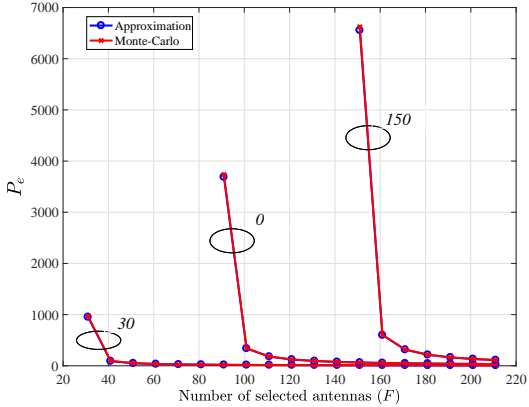}
	\caption{Comparison between the Monte Carlo simulation of $P_e$ with the approximated one in Eq.~\eqref{puldeter} for $K \in \{30, 90, 150\}$ when $\bar{R}/BW=5(bits/s/Hz)$.} \label{fig2}
\end{figure}

Next, Fig.~\ref{EED} shows the empirical eigenvalue distributions of the selected matrix of small-scale coefficients ($\tilde{\boldsymbol{H}}^*\tilde{\boldsymbol{H}}$) and its proposed approximation ($\hat{\boldsymbol{H}}^*\hat{\boldsymbol{H}}$) for $F=140$ and $K=70$.	The choice for the empirical eigenvalue distribution is due to its good reliability, and we have used a Kernel density estimator~\cite{tulino2004random} to plot Fig.~\ref{EED}. Moreover, it is known that assuming $\hat{\boldsymbol{H}}$ as an i.i.d. Gaussian matrix produces eigenvalues following a Mar\u cenko-Pastur distribution~\cite{tulino2004random}
\begin{equation}
f(x)=(1-\frac{F}{K})^+\delta(x)+\frac{K\sqrt{(x-\alpha)^+(\beta-x)^+}}{2\pi Fx\hat{\sigma}^2}, \label{MarchenkoPastur}
\end{equation}
where $\delta(x)$ is Dirac delta function, $\alpha=\hat{\sigma}^2(1-\sqrt{F/K})^2$ and $\beta=\hat{\sigma}^2(1+\sqrt{F/K})^2$. As it can be observed, the distribution of the eigenvalues of $\tilde{\boldsymbol{H}}^*\tilde{\boldsymbol{H}}$ is very similar to that of the eigenvalues of $\hat{\boldsymbol{H}}^*\hat{\boldsymbol{H}}$, showing that $\hat{\boldsymbol{H}}$ is a good approximation to $\tilde{\boldsymbol{H}}$.

\begin{figure}[t]
	\centering
	\includegraphics[width=0.5\textwidth]{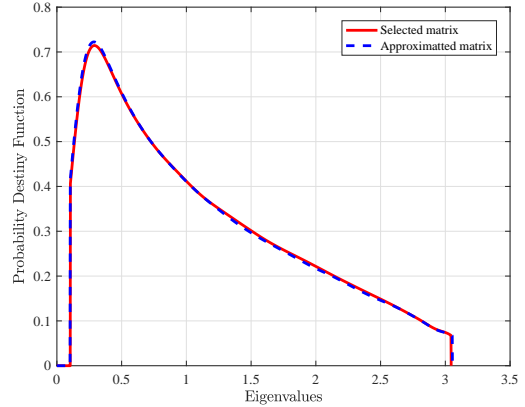}
	\caption{Empirical eigenvalue distribution of the selected matrix ($\tilde{\boldsymbol{H}}^*\tilde{\boldsymbol{H}}$) and its proposed approximation ($\hat{\boldsymbol{H}}^*\hat{\boldsymbol{H}}$), for $F=140$ and $K=70$.} \label{EED}
\end{figure}

Fig.~\ref{fig3} plots the maximal EE as a function of the number of UTs. For each $K$, the optimization of both $\bar{R}$ and $F$ is carried out and we compare the Monte Carlo simulation of the antenna selection algorithm, the approximated antenna selection algorithm of Section~\ref{approxAlg}, as well as the case when $F = M$ (without antenna selection) with the optimization of $\bar{R}$. The first observation is with respect to the tightness of our approximations in deriving~\eqref{EEd}, which is very close to the optimal curve obtained through Monte Carlo simulations. Moreover, this figure also shows a significant improvement in terms of EE when using only a subset of antennas, especially in the cases when the number of users is much lower than the number of antennas. For example, with $K=20$~UTs, the EE is about $110\%$ better than when the whole set of antennas is employed. In addition, when the number of UTs increases the optimal $F$ also increases, so that for $K > 160$ the optimum number of selected antennas is the same as the total number of available antennas; thus, both scenarios with and without antenna selection have the same performance in terms of EE.

\begin{figure}[t]
	\centering
	\includegraphics[width=0.5\textwidth]{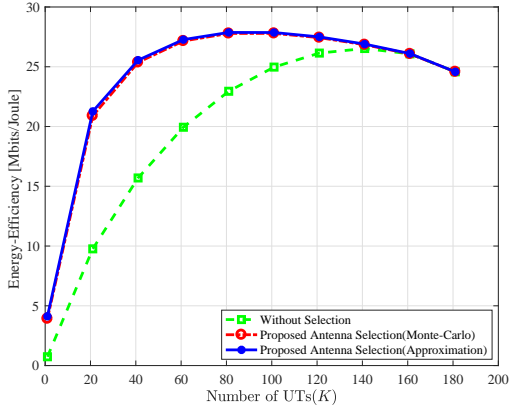}
	\caption{EE of the Monte Carlo simulation for the antenna selection algorithm with optimal $F$, the approximated antenna selection algorithm with optimal $F$, as well as the case when $F = M$ (without antenna selection). Moreover, $\bar{R}$ is also optimized for each scenario.} \label{fig3}
\end{figure}

Fig.~\ref{MMSEMF} represents the EE for MF and MMSE, both of them obtained through Monte Carlo simulation. As it can be seen comparing both results, EE is improved up to $60\%$ using the proposed antenna selection with  MF (compared to the case without antenna selection), and up to $88\%$ using the proposed antenna selection with MMSE. However, a big difference is observed when comparing the magnitude of the EE for MF and MMSE, which is mainly caused by a larger amount of interference in the case of MF. From our simulations we also notice that, due to the higher interference, MF tends to select the number of antennas equal to $K$, even for high number of UTs, jeopardizing the EE. On the other hand, MMSE shows very similar performance compared to ZF, as suggested in Section~\ref{System Model}, with MMSE achieving slightly better EE, while ZF exhibits better improvement when using the proposed antenna selection scheme. Nevertheless, when EE achieves it optimal point, the corresponding data rates and number of selected antennas for ZF and MMSE are very close, while MF has at least $3$ times lower data rates compared to the other schemes.

Fig.~\ref{fig6} shows EE for different values of $\bar{R}$, fixed at $40$~Mbits/s, $70$~Mbits/s and $100$~Mbits/s, respectively, so that the optimization is done with respect to the number of selected antennas, $F$, only in ZF case. This scenario represents many practical cases when the operators are interested in setting a constant data rate for the users. As it can be observed, there is a larger gap between the EE of the proposed antenna selection scheme and employing all $M$ antennas for lower data rates, which is because the optimum number of selected antennas is small for lower data rates, consequently increasing the EE. On the other hand, by increasing the required data rate, the optimal number of antennas $F^\star$ also increases until it reaches the saturation point of being equal to $M$.

\begin{figure}[t]
	\centering
	\includegraphics[width=0.5\textwidth]{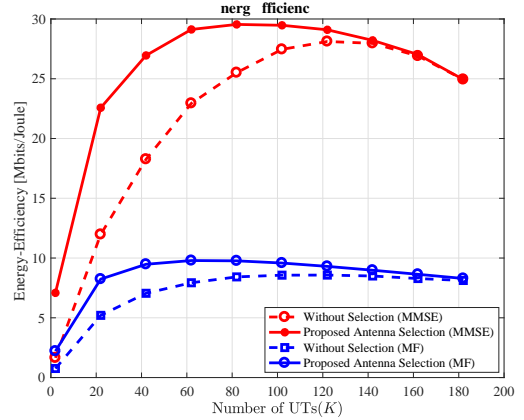}
	\caption{EE of Monte Carlo simulation for the proposed antenna selection algorithm and without selection case for MMSE and MF. Optimization is done for both $F$ and $\bar{R}$.} \label{MMSEMF}
\end{figure}

\begin{figure}[t]
	\centering
	\includegraphics[width=0.5\textwidth]{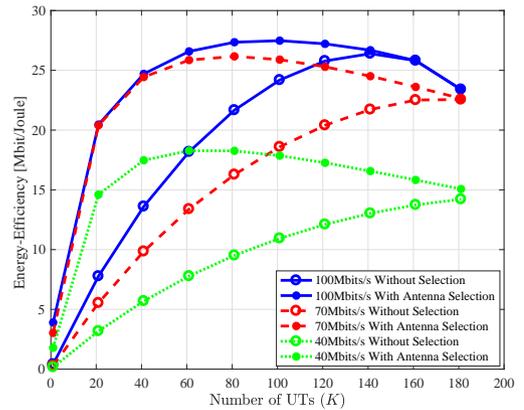}
	\caption{EE of the Monte Carlo simulation for the antenna selection algorithm with optimal $F$, the approximated antenna selection algorithm with optimal $F$, as well as the case when $F = M$ (without antenna selection). The data rate is assumed to be fixed at $\bar{R} = 40, 70$ and $100~Mbits/s$.} \label{fig6}
\end{figure}

For the same scenarios as Fig.~\ref{fig6}, the optimal number of selected antennas is shown in Fig.~\ref{fig7} as a function of the number of UTs. As we can notice, $F^\star$ is always slightly higher than $K$, but the relation depends on the requirement in terms of $\bar{R}$. Moreover, $F^\star$ saturates at the maximum value of $220$ (equal to $M$) for higher data rates, when $\bar{R} = 100$~Mbps in this case, for which the EE is the same as in the non-selection scheme. Let us remark that we assume that $K$ is not a parameter to be optimized in this work, since it depends on the particular network scenario. Thus, Fig.~\ref{fig7} provides $F^\star$ as a function $K$.

\begin{figure}[t]
	\centering
	\includegraphics[width=0.5\textwidth]{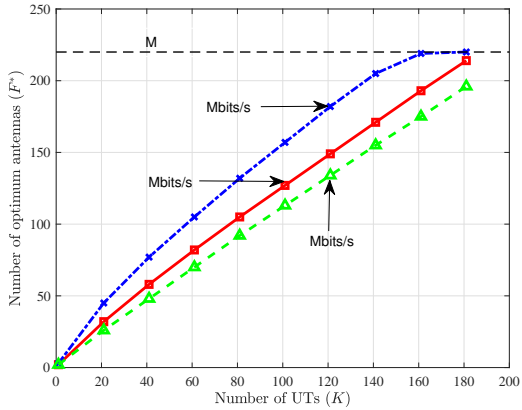}
	\caption{Optimum number of selected antennas ($F^\star$) as a function of the number of UTs ($K$) for different values of $\bar{R}$} \label{fig7}
\end{figure}

Finally, Fig.~\ref{fig8} investigates the EE as a function of the system spectral efficiency (defined as $\bar{R}/BW$ in bits/s/Hz \cite{ngo2013energy}), in the case when the number of UTs is $K=90$ and $F$ is optimized in order to maximize the EE. A remarkable improvement in terms of EE can be observed with the proposed antenna selection scheme. Moreover, the gap between using antenna selection or not using it, decreases for high spectral efficiency, which is due to the requirement for more antennas to serve higher data rates while there exists a limited number of available antennas. It is worth noting that more than $30\%$ of improvement in terms of EE can be achieved when the system operates either with $3.7$~bits/s/Hz or $2$~bits/s/Hz, which are the spectral efficiency targets for LTE-advanced systems at the downlink and uplink, respectively~\cite{rumney2013lte}.

\begin{figure}[t]
	\centering
	\includegraphics[width=0.5\textwidth]{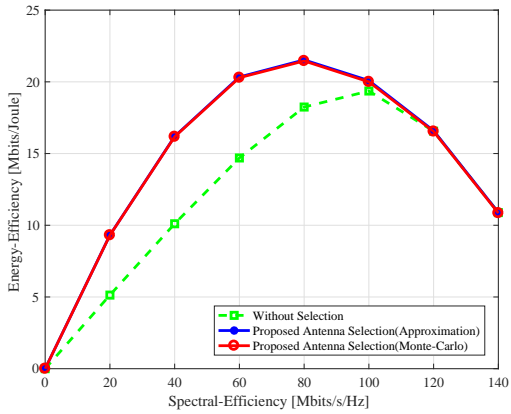}
	\caption{EE as a function of the spectral efficiency. For each point of the curve, the optimization is performed in terms of $F$.} \label{fig8}
\end{figure}

\section{Conclusion} \label{Conclusion}
Massive MIMO systems yield high data rates to the user terminal by employing a very large number of antennas. However, as a consequence, the power consumption also increases with the number of antennas, while EE is a vital concern for new generation wireless systems. In this paper we considered a multi-user massive MIMO scenario and we proposed an antenna selection scheme so that the system may operate with a subset of $F \leq M$ antennas in order to maximize the EE. By employing a realistic power consumption model for massive MIMO, we derived a closed-form approximation for the EE formulation ir order to reduce the mathematical complexity, which is shown to be very tight with the optimal results. Simulation results showed that using antenna selection significantly improves the EE, especially when the number of UTs is relatively small with respect to the number of available antennas. Moreover, the optimal number of antennas was shown to be slightly higher than the number of UTs, while the relation depends on the required data rate. Finally, when both the number of employed antennas and the average data rate to the users are optimized, up to $110\%$ for $K=20$ improvement in EE is observed.

\appendix
\section{Proof of Corollary \ref{cor1}} \label{ApCor1}
Using a triangle inequality in~\eqref{thm1_1} we have that
\begin{align}
&\sum_{i=1}^{M}\lambda_i\mu_{i:M} - \left\lvert\sum_{i=1}^{M}\lambda_i\mu\right\lvert \le \left\lvert\sum_{i=1}^{M}\lambda_i\mu_{i:M}\right\lvert - \left\lvert\sum_{i=1}^{M}\lambda_i\mu\right\lvert \nonumber\\& \le \left\lvert\sum_{i=1}^{M}\lambda_i(\mu_{i:M}-\mu)\right\lvert \le \sigma\sqrt{M\sum_{i=1}^{M}(\lambda_i-\bar{\lambda})^2}.
\end{align}
Then, noting that
\begin{equation}
\mu=2K\sigma^2>0\hspace{3mm}, \hspace{3mm}\sigma_z^2=4K\sigma^4, \nonumber
\end{equation}
we would have
\begin{equation} \label{aftertri}
\sum_{i=1}^{M}\lambda_i\mu_{i:M}\le \sum_{i=1}^{M}\lambda_i2K\sigma^2+2\sigma^2\sqrt{KM\sum_{i=1}^{M}(\lambda_i-\bar{\lambda})^2}.
\end{equation}

Now let us choose $\lambda_i$ as the following
\begin{equation} \label{zabet}
\lambda_i= \left\{
\begin{array}{ll} 
1 &  \hspace{5mm}\text{if the $i^\text{th}$ antenna is selected }\\
0 &  \hspace{5mm}\text{otherwise},
\end{array} \right.
\end{equation}
so that $\bar{\lambda}=\frac{F}{M}$.

By replacing~\eqref{zabet} into the right hand side of~\eqref{aftertri} we have
\begin{equation}
\begin{split}
&\sum_{i=1}^{M}\lambda_i2K\sigma^2+2\sigma^2\sqrt{KM\sum_{i=1}^{M}(\lambda_i-\bar{\lambda})^2} \\
&=2FK\sigma^2+2\sigma^2\sqrt{KM\left[F\left(1-\frac{F}{M}\right)^2+(M-F)\left(-\frac{F}{M}\right)^2\right]} \\
&=2FK\sigma^2+2\sigma^2\sqrt{K[F(M-F)]} \label{eq38}.
\end{split}
\end{equation}

Finally, inserting~\eqref{eq38} into~\eqref{aftertri}, replacing $\lambda_i$ by $E\{a_{f:M}\}$ and dividing both sides of the inequality by $F$ yields \eqref{inequality1}, concluding the proof.
\qed

\section{Proof of Corollary \ref{cor2}} \label{ApCor2}
Taking into account that $\boldsymbol{A}^{(dl)}$ is a diagonal matrix we can write
\begin{equation}
\mathbb{E}\{\boldsymbol{1}_K^T (\boldsymbol{A}^{(dl)})^{-1}\boldsymbol{1}_K\}=\mathbb{E}\{tr(\boldsymbol{A}^{(dl)})^{-1}\}.
\end{equation}
Then, using~\eqref{Ak,k} and noting that $\parallel\boldsymbol{v}_k\parallel^2=(\tilde{\boldsymbol{G}}^H\tilde{\boldsymbol{G}})^{-1}_{(k,k)}$ we have
\begin{equation}
\mathbb{E}\{tr(\boldsymbol{A}^{(dl)})^{-1}\}=(e^{\frac{\bar{R}}{BW}}-1)\mathbb{E}\{tr(\tilde{\boldsymbol{G}}^H\tilde{\boldsymbol{G}})^{-1}\}.
\end{equation}

Moreover, using~\eqref{wishart} in \cite[eq.(50)]{maiwald2000calculation},
\begin{equation}
\mathbb{E}\{tr(\boldsymbol{A}^{(dl)})^{-1}\}\approx\frac{(e^{\frac{\bar{R}}{BW}}-1)}{(F-K)(1+\sqrt{\frac{M-F}{FK}})}\mathbb{E}\{tr(\boldsymbol{B}^{-1})\}.
\end{equation}

Since the expectation is the same for all UTs due to the assumed path-loss model~\cite{bjornson2014optimal}, then $\mathbb{E}\{r_k^{-1}\}=\mathbb{E}\{r^{-1}\}$, so that
\begin{equation}
\mathbb{E}\{tr(\boldsymbol{B}^{-1})\}=K\mathbb{E}\{r^{-1}\}
\end{equation}
and therefore
\begin{equation} \label{EAdl}
\mathbb{E}\{tr(\boldsymbol{A}^{(dl)})^{-1}\}\approx\frac{K(e^{\frac{\bar{R}}{BW}}-1)}{(F-K)(1+\sqrt{\frac{M-F}{FK}})}\mathbb{E}\{r^{-1}\}.
\end{equation}

Finally, putting~\eqref{EAdl} into~\eqref{Peup} results in~\eqref{puldeter}, which concludes the proof.
\qed

\end{document}